\documentclass[sigconf,screen]{aamas}

\usepackage{amsmath,amsfonts}
\newtheorem{remark}{Remark}
\usepackage{algorithm}
\usepackage{array}
\usepackage[caption=false,font=normalsize,labelfont=sf,textfont=sf]{subfig}
\usepackage{textcomp}
\usepackage{stfloats}
\usepackage{url}
\usepackage{verbatim}
\usepackage{graphicx}
\usepackage{booktabs}
\usepackage{color,soul}

\usepackage{algpseudocode}
\usepackage{caption} 
\usepackage{makecell}

\newtheorem{theorem}{Theorem}        
\newtheorem{lemma}{Lemma}            
\newtheorem{proposition}{Proposition}
\newtheorem{corollary}{Corollary}

\usepackage{balance} 
\usepackage{xcolor}
\usepackage{hyperref}
\hypersetup{colorlinks=true, allcolors=green}

\setcopyright{ifaamas}
\acmConference[AAMAS '26]{Proc.\@ of the 25th International Conference
on Autonomous Agents and Multiagent Systems (AAMAS 2026)}{May 25 -- 29, 2026}
{Paphos, Cyprus}{C.~Amato, L.~Dennis, V.~Mascardi, J.~Thangarajah (eds.)}
\copyrightyear{2026}
\acmYear{2026}
\acmDOI{}
\acmPrice{}
\acmISBN{}

\title[AAMAS-2026 Formatting Instructions]{Peer-Aware Cost Estimation in Nonlinear General-Sum Dynamic Games for Mutual Learning and Intent Inference}

\author{Seyed Yousef Soltanian}
\affiliation{
  \institution{Arizona State University}
  \city{Tempe}
  \country{United States}}
\email{ssoltan2@asu.edu}

\author{Wenlong Zhang}
\affiliation{
  \institution{Arizona State University}
  \city{Tempe}
  \country{United States}}
\email{wenlong.Zhang@asu.edu}

\begin{abstract}
Dynamic game theory is a powerful tool in modeling multi-agent interactions and human-robot systems. In practice, since the objective functions of both agents may not be explicitly known to each other, these interactions can be modeled as incomplete-information general-sum dynamic games. Solving for equilibrium policies for such games presents a major challenge, especially if the games involve nonlinear underlying dynamics.
To simplify the problem, existing work often assumes that one agent is an expert with complete information about its peer, which can lead to biased estimates and failures in coordination. To address this challenge, we propose a nonlinear peer-aware cost estimation (N-PACE) algorithm for general-sum dynamic games. In N-PACE, using iterative linear quadratic (ILQ) approximation of dynamic games, each agent explicitly models the learning dynamics of its peer agent while inferring their objective functions and updating its own control policy accordingly in real time, which leads to unbiased and fast learning of the unknown objective function of the peer agent. Additionally, we demonstrate how N-PACE enables intent communication by explicitly modeling the peer's learning dynamics. Finally, we show how N-PACE outperforms baseline methods that disregard the learning behavior of the other agent, both analytically and using our case studies.\footnotemark Codes: \href{https://github.com/YousefSoltanian/AAMAS2026-NPACE}{AAMAS2026-NPACE}.

\end{abstract}


\newcommand{\BibTeX}{\rm B\kern-.05em{\sc i\kern-.025em b}\kern-.08em\TeX}

\begin{document}


\pagestyle{fancy}
\fancyhead{}


\maketitle 


\section{Introduction}

General-sum dynamic games offer a powerful framework for modeling multi-agent interactions \cite{schwarting2019social, losey2018review}, with applications in autonomous driving \cite{chen2021shall}, assistive robots \cite{li2019differential}, and social navigation \cite{fridovich2020efficient}. However, the process of finding the feedback Nash equilibrium (FBNE) solution of such games, which can be achieved by solving a set of Hamilton-Jacobi-Bellman (HJB) equations \cite{bressan2010noncooperative}, is inherently complex, as traditional dynamic programming approaches suffer from the \textit{curse of dimensionality} \cite{powell2007approximate}. 
This problem becomes even more challenging when agents do not know each other's objective function and must perform inference online while making decisions, resulting in an incomplete-information differential game, a common scenario in human-robot interaction (HRI) \cite{li2019differential, nikolaidis2017game, kedia2024interact, tian2023towards} and autonomous vehicle interaction \cite{amatya2022shall, wang2019enabling, xing2019driver}.

\begin{figure}[t]
\centering
\includegraphics[width=\columnwidth]{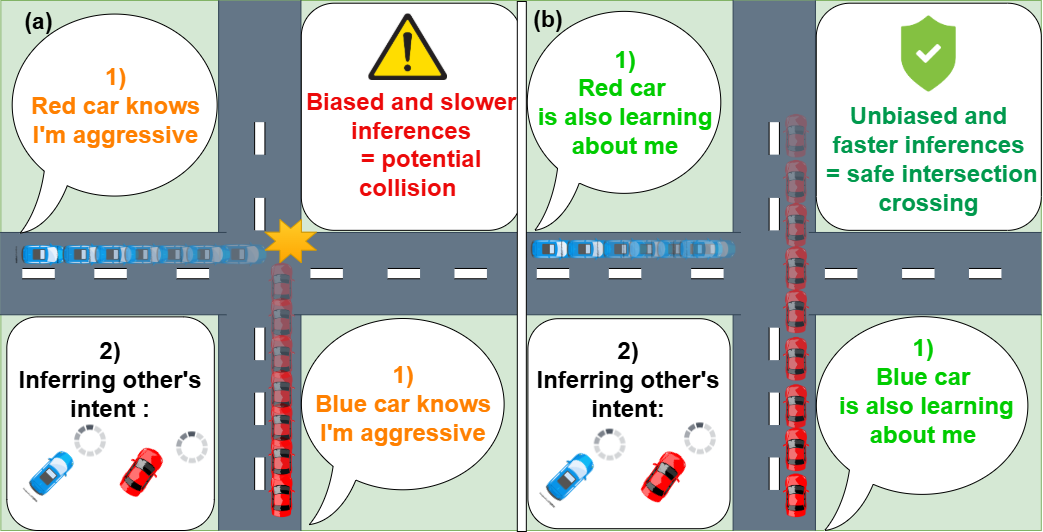}
\caption{Demonstration of the core idea in N-PACE using a simulated
general-sum game between two autonomous cars at an intersection. In
(a), each agent assumes its peer knows its aggressiveness level, leading to biased inferences and possible collisions.
In (b), N-PACE models the peer agent as a learning entity, resulting in a safe intersection crossing.
}
\label{fig1}
\end{figure}

\begingroup
\renewcommand\thefootnote{}\footnotetext{\textbf{Extended preprint version.} This manuscript is an extended version of our paper accepted to \emph{Proc.\@ of AAMAS 2026}.}
\endgroup
\setcounter{footnote}{0}

Existing approaches to solve such games often rely on the assumption that one agent is an \textit{expert} which knows the true intent of other agents \cite{peters2024contingency,le2021lucidgames, li2024intent, schwarting2019social}. Despite success of these approaches in situations where the peer agent is indeed a complete information agent, in scenarios where all agents are learning simultaneously and inferring each other’s objective functions (which is defined as their \textit{intent} in this paper), these assumptions can lead to biased estimates, poor coordination, and failure to achieve each agent's goal \cite{liu2016blame,soltanian2025pace}.
This is because while doing the inference, they fail to account for the fact that other agents can be learners, whose policies are conditioned on previously observed data \cite{kedia2024interact}. It is even more important to consider such mutual learning in a competitive setting where an agent has no information about the others. 

In this paper, we highlight the importance of accounting for the learning dynamics of peer agents while inferring their objectives, addressing critical limitations in existing approaches. We achieve this goal by developing a peer-aware cost estimation algorithm for general-sum nonlinear games (N-PACE), where each agent leverages prior knowledge about its peer agent’s learning dynamics (e.g., gradient-based learner, Bayesian learner) to condition its current inference of the other's objective on their parameter learning, thereby removing the bias in their inferences. The N-PACE algorithm employs iterative linear quadratic (ILQ) approximation of the nonlinear general-sum game \cite{fridovich2020efficient}, both in policy update and intent inference phases. We also demonstrate how (ILQ) enables efficient approximation during the inference phase of our algorithm.

\noindent\textbf{Contributions.} This paper makes the following contributions.  
(1) We introduce the N-PACE algorithm for mutual intent inference in multi-agent interaction with continuous action-state and continuous intent space, where all agents are making online inferences. 
(2) We empirically and theoretically demonstrate the advantages of N-PACE in convergence, safety, and task completion, compared to baseline methods that assume the peer agent is an expert who knows the ego agent's true intent.
(3) We demonstrate how N-PACE, along with access to the learning dynamics of the peer agent, can be leveraged for intent communication and signaling to facilitate mutual learning and adaptation.

The closest works to ours are \cite{fridovich2020efficient,li2024intent,liu2016blame,soltanian2025pace}, where \cite{fridovich2020efficient} uses ILQGame for solving \textit{complete information} multi-agent interactions. In \cite{li2024intent}, the ILQ approach is employed in a general-sum game where an ``expert'' demonstrates its intent to a learner agent knowing their learning dynamics, but it does not consider the situation when two learning agents (no experts) try to coordinate which is a more practical setting in multi-agent interactions. Finally, while \cite{liu2016blame,soltanian2025pace} study the problem of mutual learning and adaptation, the former's problem setup is limited to a single-step optimization horizon cost optimization, and the latter considers only the linear quadratic setup. This work extends those to nonlinear general-sum games, which is a more practical multi-agent setup.


\section{Related Work}

\label{sec: related}
\textbf{Solving general-sum games.} Multi-agent interactions have been explored using various frameworks, including multi-agent reinforcement learning \cite{lowe2017multi, canese2021multi}, adaptive control \cite{chen2019control}, and game theory \cite{yang2020overview, bloembergen2015evolutionary}. A large portion of these interactions can be modeled as general-sum dynamic games \cite{zhang2023approximating, zhang2024pontryagin}.
Traditional approaches for solving such games rely on adaptive dynamic programming to compute feedback Nash equilibria (FBNE) \cite{vamvoudakis2011multi}. However, these methods are often hindered by the \textit{curse of dimensionality} \cite{powell2007approximate}, making them impractical for high-dimensional problems and real-time implementation in robotics. To overcome these limitations, physics-informed neural networks (PINN) have been used to solve the HJB equation more efficiently \cite{zhang2024pontryagin}. Some recent work has focused on iterative linear-quadratic  (ILQGames) approximations, which offer a more computationally efficient strategy for approximating the multi-agent interactions solution by finding the \textit{local Nash equilibrium} of such games \cite{fridovich2020efficient}. In this paper, we employ ILQGame for both the learning and control phases to find the approximated Nash equilibrium policies.

\textbf{Incomplete information general-sum games.} 
 Solving incomplete information general-sum games is significantly more challenging than their complete-information counterparts \cite{aumann1995repeated, cardaliaguet2012games}. These games often require simplifications, such as modeling one agent (the ego agent) as a learner while assuming others are fully informed \cite{peters2024contingency,le2021lucidgames, laine2021multi}, treating the problem as an adaptive control problem \cite{li2019differential}, discretizing the intent space or relying on offline datasets for intent inference and belief updates \cite{hu2023emergent, mehr2023maximum, chen2021shall}.  
Our work is motivated by the lack of studies on incomplete information dynamic games where all agents are learners \cite{jacq2019learning}. We will highlight the potential issues caused by these simplifications, such as modeling other agents as fully informed.

\textbf{Learning from a learner.}
Online inference of an expert agent's objective function has been extensively studied through inverse optimal control and inverse reinforcement learning \cite{molloy2020online, rhinehart2017first, self2022model}. While inverse non-cooperative games have been explored both offline and online \cite{molloy2022inverse, li2023cost}, existing works focus on learning from an expert system.  
Some works have considered \textit {opponent shaping} and \textit{learning from a learner} agent by incorporating their learning dynamics into offline multi-agent policy training approaches \cite{foerster2017learning, jacq2019learning}, but they are inherently offline multi-agent policy training algorithms and not primarily designed for real-time intent inference and objective learning. Moreover, they are not built on game-theoretic setups where the future action of one agent is conditioned on others. But motivated by these \textit{learning from a learner} approaches and recent efforts in modeling human learning dynamics in human-robot interactions \cite{tian2023towards, kedia2024interact}, we propose a framework for mutual intent inference in systems where all agents are learners.

\section{Problem Formulation} 

\textbf{Notations and Assumptions.} We focus on a two-player setting with agents \( i \) and \( j \). When referring to an arbitrary agent, we use the index \( k \) and denote their counterpart as \( -k \). We consider the following system dynamics
\begin{equation}
s_{t+1} = f_t(s_{t}, a^i_t, a^j_t),  
\label{eq:dynamic}
\end{equation}
where \( s_t \in \mathbb{R}^n \) is the shared state vector at time step \( t \), and \( a^i_t, a^j_t \in \mathbb{R}^m \) are the control inputs of agents \( i \) and \( j \), respectively. The function \( f_t \) is assumed to be differentiable with respect to \( s_t, a^i_t, \) and \( a^j_t \). We assume that all states and control signals \( s_t, a^i_t, a^j_t \) are observable to both players at each time step.  
For a fixed time horizon \( T \), the running cost of each player \( k \in \{i,j\} \) is defined as  
$g^k_t(s_t, a^i_t, a^j_t; \theta^k)$,  
where \( g^k_t \) is a scalar function that is twice differentiable with respect to \( s_t, a^i_t, a^j_t \) and differentiable with respect to the \textit{intent} parameter \( \theta^k \) for \( g^k_t \) (e.g., the targeted goal state or a weighting factor).  
Additionally, we use the notation \( \langle \hat{\cdot} \rangle \) to denote the parameter estimate.

\textbf{General-sum game formulation.} We consider a receding horizon cost minimization problem for each agent, where at each time step \( t \in \{0,1,\dots,T\} \), each player \( k \) aims to minimize the cumulative cost:
\begin{equation}
     J^k_t(s_t, a^i_t, a^j_t) = \sum_{\tau=t}^{T} g^k_\tau(s_\tau, a^i_\tau, a^j_\tau; \theta^k). 
     \label{eq:J}
\end{equation}
The Feedback Nash Equilibrium (FBNE) \cite{bacsar1998dynamic} is a solution concept for the above minimization problem, constrained by (\ref{eq:dynamic}), where neither agent has an incentive to deviate from their policy, as doing so would increase their cumulative cost. In an FBNE, each agent’s policy is a function of the current state, ensuring that at every time step, the chosen control actions are optimal given the strategies of the other agents. We represent these policies using $\pi^k_t(s_t; \theta^i, \theta^j) $ such that \( a^k_t = \pi^k_t(s_t; \theta^i, \theta^j) \), emphasizing that they depend on both \( \theta^i \) and \( \theta^j \), which is crucial for further formulation in our work.

\textbf{The incomplete information game.}
The previously discussed FBNE solution was derived under the assumption that both agents have full knowledge of \( \theta^i \) and \( \theta^j \). However, in incomplete information games, \( \theta^j \) may be unknown to agent \( i \), and \( \theta^i \) may be unknown to agent \( j \). As a result, at each time step, they must rely on estimations \( \hat{\theta}^j_t \) and \( \hat{\theta}^i_t \), respectively, and update these estimates over time. This leads to policies of the form \( \pi^i_t(s_t; \theta^i, \hat{\theta}^j_t) \) and \( \pi^j_t(s_t; \hat{\theta}^i_t, \theta^j) \), which are not necessarily the same as FBNE strategies. To address this problem, the next section introduces how we use the ILQ method to obtain the policy \( \pi^k_t(s_t; \theta^k, \hat{\theta}^{-k}_t) \) for each agent \( k \), based on their current estimate of the other agent's parameter \( \hat{\theta}^{-k}_t \). We then propose appropriate update rules for these estimates in Sec. \ref{sec:learning}.

\subsection{ILQGame for Multi-Agent Interactions}
\label{sec: ILQR}
The iterative linear quadratic regulator (iLQR) approach for locally optimal feedback control of nonlinear systems was first introduced in \cite{li2004iterative}. More recent work, such as \cite{fridovich2020efficient}, has demonstrated its efficiency in computing local Nash equilibrium solutions for general-sum dynamic games in real time. In our incomplete information game, where we need policies \( \pi^k_t \) to be computed as functions of the estimated parameters \( \hat{\theta}^i_t \) and \( \hat{\theta}^j_t \) in real time, we leverage the proposed ILQGame algorithm in \cite{fridovich2020efficient} to solve the game for each agent \( k \) at every time step. The main reasons for using ILQGame in this work are its efficiency in real-time computations and, more importantly, its capability to parameterize the obtained policies as functions of \( \hat{\theta}^i_t \) and \( \hat{\theta}^j_t \), which is crucial for having a differentiable layer in N-PACE as shown in Sections \ref{sec:learning} and \ref{sec: N-PACE}.

In the ILQGame approach, we begin with an initial admissible policy for each agent. In the forward pass, by applying these policies to the dynamics (\ref{eq:dynamic}), we obtain a trajectory \( \boldsymbol{\zeta^0_t} = \{(s^0_{\tau}, a^{i0}_{\tau}, a^{j0}_{\tau}) \mid \tau \in \{t, t+1, \dots, T\} \} \). Next, by linearizing the dynamics \( f_t \) and quadratizing the running costs \( g^k_t \) in the neighborhood of the trajectory \( \boldsymbol{\zeta^0_t} \), we construct a finite-horizon LQ game, which can be solved via coupled Riccati equations using the methods proposed in \cite{bacsar1998dynamic}. The linear feedback control solution of this LQ game is then used to update the agents' policies, as suggested in \cite{fridovich2020efficient}. This iterative process continues until the system trajectory remains unchanged at each iteration and converges to a fixed trajectory, known to be the local Nash equilibrium trajectory \( \boldsymbol{\zeta^*_t} \).

In an incomplete information game, each agent learns  the other's parameter \( \theta^{-k} \), resulting in an estimate \( \hat{\theta}^{-k}_t \). As a result, each agent independently solves a different ILQ problem to update their policy using \( \theta^k \) and \( \hat{\theta}^{-k}_t \), leading to a distinct converged trajectory \( \zeta^{*k}_t(\theta^k, \hat{\theta}^{-k}_t) \) at each step. At the converged trajectories, each agent's policy consists of a feedforward time-varying term $\alpha_t$ and a feedback time varying term $P_t$ in the form of
\begin{equation}
\begin{aligned}
    \pi^k_t(s_t; \theta^k, \hat{\theta}^{-k}_t) &= -\alpha_t(\theta^k, \hat{\theta}^{-k}_t) \\
    &\quad - P_t(\theta^k, \hat{\theta}^{-k}_t)(s_t - s^{*k}_t(\theta^k, \hat{\theta}^{-k}_t)),
\end{aligned}
\label{eq:policy}
\end{equation}
 where \( s^{*k}_t(\theta^k, \hat{\theta}^{-k}_t) \) represents states stored in the converged trajectory. In (\ref{eq:policy}), we explicitly maintain and emphasize the dependence of each parameter on \( \theta^k \) and \( \hat{\theta}^{-k}_t \). This not only highlights how agents update their policies based on new estimates but also serves as an introduction to the N-PACE algorithm.

\begin{remark}{(Extension to Other Methods)}
The core idea in this paper is not limited to using ILQ for solving general-sum games. The key requirement for the success of N-PACE is to have a policy explicitly parameterized by \( \theta^k \) and \( \theta^{-k} \). In fact, alternative approaches, such as some reinforcement learning algorithms \cite{williams1992simple, lowe2017multi}, can also borrow the N-PACE core idea as long as they can generate policies parametrized by intent parameters.
\end{remark}

\subsection{Modeling the Learning Dynamics}
\label{sec:learning}
In Sec. \ref{sec: ILQR}, we discussed how an agent \( k \) can update its policy at each step based on the current estimate \( \hat{\theta}^{-k}_t \) to solve a general-sum game. However, the question of how to learn and update \( \hat{\theta}^{-k}_t \) remains open.  
In this section, we propose a learning algorithm under the assumption that agent \( -k \) is an expert, for the sake of clarity and brevity. We then introduce approximation methods to enhance the efficiency of this learning process for real-time applications.  
Finally, in  Sec. \ref{sec: N-PACE}, we discuss how this learning approach can be extended to N-PACE for scenarios where all agents are learners.

\textbf{The learning dynamics $h^k_l$.} When agents can observe the states and control signals \( s_t, a^k_t, a^{-k}_t \) at each step, they can leverage this information to construct an error term for the online learning of the parameters \( \hat{\theta}^{-k}_t \).  
An agent \( k \), knowing that their peer is solving an ILQGame to generate policies, can attempt to model this policy using their current estimate \( \hat{\theta}^{-k}_t \) as \( \hat{a}^{-k}_t = \pi^{-k}_t(s_t; \hat{\theta}^{-k}_t, \theta^k) \) as given in (\ref{eq:policy}). Consequently, they can define the following update rule as their learning dynamics \( h^k_l \) to refine the estimate of \( \hat{\theta}^{-k}_t \)
\begin{equation}
\label{eq:learning_dy}
\begin{aligned}
    \hat{\theta}^{-k}_{t+1} &= h^k_l(\hat{\theta}^{-k}_t, \theta^k, a^{-k}_t, s_t) 
    = \hat{\theta}^{-k}_t \\&- \alpha_k (\pi^{-k}_t(s_t;\hat{\theta}^{-k}_t,\theta^k) - a^{-k}_t)^\top 
    \nabla_{\hat{\theta}^{-k}_t} \pi^{-k}_t(s_t;\hat{\theta}^{-k}_t,\theta^k)
\end{aligned}
\end{equation}

\begin{remark}{(Expert Peer Assumptions)}
Notice that in  (\ref{eq:learning_dy}), when modeling the policy of the other agent \( -k \), the parameter \( \theta^k \) is included in the policy model. This effectively assumes that the other agent has full knowledge of \( \theta^k \) and is treated as an expert, which can lead to incorrect and biased estimates.
\end{remark}

\begin{remark}{(Flexibility of Learning Dynamics)} The proposed learning dynamics \( h^k_l \) follows a simple gradient descent algorithm with a learning rate \( \alpha_k \), but it is not restricted to gradient-based methods in the N-PACE. It can also take the form of a Bayesian learner or even a neural network model, as explored in \cite{tian2023towards}.  
\end{remark}
\textbf{Efficient approximation of the gradient.} The learning dynamics proposed in (\ref{eq:learning_dy}) require calculating the gradient of the policy model with respect to the unknown parameter \( \hat{\theta}^{-k}_t \). The ILQ method, which generates the policy at each step, is an iterative process where the generated trajectories and policy improvement steps depend on the unknown parameter \( \hat{\theta}^{-k}_t \). While this approach still relates the policy to \( \hat{\theta}^{-k}_t \), the iterative nature of ILQ can make gradient calculation computationally expensive or even intractable for real-time implementation.  
Notably, in (\ref{eq:policy}), as well as in the converged trajectory \( \zeta^{*k}_t(\theta^k, \hat{\theta}^{-k}_t) \), all parameters are explicitly expressed as functions of \( \hat{\theta}^{-k}_t \). 

To address this challenge, we adopt the approximation method proposed in \cite{li2023cost}. Specifically, we ignore the dependency of the converged trajectory in ILQ on \( \hat{\theta}^{-k}_t \) and replace the gradient of the policy with the gradient of the linear policy obtained by solving the LQ game around the converged trajectory of iLQR.  
The intuition behind this approximation is that in ILQ, the policy derived from solving the expanded LQ game around the converged trajectory serves as a good approximation of the final converged policy itself, as discussed and demonstrated in \cite{li2023cost}. However, automatic differentiation (AD) methods can still be used to compute this gradient if there are no computational constraints. Alternatively, if a different game solver allows the policy to be directly expressed as a function of the unknown parameter \( \hat{\theta}^{-k}_t \) (e.g., an actor network receiving the parameter as a state input), one can bypass the need for this approximation method.

\subsection{N-PACE for Multi-Agent Interactions}
\label{sec: N-PACE}
In Sec. \ref{sec:learning}, we proposed the learning dynamics \( h^k_l \) for each agent. However, the update rule in (\ref{eq:learning_dy}) assumes that agent \( -k \) is aware of \( \theta_k \) while generating its policy, effectively treating agent \( -k \) as an expert during estimation. This assumption is clearly invalid when all agents are learning.  
In N-PACE, we address this issue by ensuring that each agent \( k \) not only learns parameters of the other agent, but also accounts for the learning procedure of the other agent \( -k \). This is feasible when agents have access to both learning dynamics, \( h^i_l \) and \( h^j_l \). Consequently, they can replace \( \theta_k \) with \( \hat{\theta}^k_t \) in equation (\ref{eq:learning_dy}), eliminating the bias in the learning process.  
As a result, each agent, by having access to learning dynamics and tracking the learning of the other agent, arrives at the following system of learning equations at each step
\begin{equation}
\label{eq: PACE}
\begin{aligned}
    \hat{\theta}^{i}_{t+1} &= h^i_l(\hat{\theta}^{i}_t, \hat{\theta}^{j}_t, a^{i}_t, s_t), \\
    \hat{\theta}^{j}_{t+1} &= h^j_l(\hat{\theta}^{j}_t, \hat{\theta}^{i}_t, a^{j}_t, s_t).
\end{aligned}
\end{equation}

Details of N-PACE are shown in Algorithm 1. It should be noted that N-PACE does not require both agents to follow the same learning dynamics; what matters in N-PACE is that agents have knowledge of each other's learning dynamics. An important property of N-PACE is that when both agents start with the same initial estimates for \( \hat{\theta}^{i}_{0} \) and \( \hat{\theta}^{j}_{0} \) and continuously track each other's learning, they will always track the same estimations for both parameters. As a result, the learning dynamics (\ref{eq: PACE}) are effectively centralized, even though each agent performs the learning process independently.  
The situation where agents know others’ initial priors/estimates often happens in multi-agent settings when there is a strong prior on the parameters (e.g., a predefined level of rationality or aggressiveness in autonomous driving) or when agents are designed to coordinate despite having different unknown intentions.  
\begin{remark}(Unknown learning dynamics)
    If the initial estimates or learning dynamics of each agent are uncertain (such as human-robot interactions), we refer the readers to works such as \cite{tian2023towards, xie2021learning, wang2022influencing}, which focus on data-driven modeling of these learning processes or latent representation of such learning behaviors.
\end{remark} 
Even if agents misspecify each other’s initial estimates, N-PACE tends to attenuate the resulting bias compared to modeling the peer as a complete information agent that introduces a fixed bias into the learning procedure. Intuitively, for example, when the peer’s update map admits a stable fixed point, the coupled updates steer the estimates toward that fixed point, reducing the initial mismatch over time, trying to keep N-PACE robust to the initial mismatch. In contrast, methods that approximate the peer as a fully informed “expert” introduce a \emph{structural} bias: the model mismatch is baked into the inference mechanism and therefore persists throughout the interaction \cite{liu2016blame,soltanian2025pace}. A robustness analysis of N-PACE to misspecified learning dynamics and initial estimates has been provided in Sec ~\ref {robust}.
While N-PACE’s advantage in eliminating the inference bias of baseline methods is intuitive, the problem it addresses is generally non-convex; thus, any convergence guarantee necessarily depends on the stage-cost structure and the underlying dynamics. But to make the advantage theoretically more precise—relative to baselines that treat the peer as an “expert”—we analyze a class of dynamic games in which each agent’s policy is linear with respect to its own intent parameter $\theta^k$. This structure arises when the Hessian of the stage cost $g_k(\cdot)$ with respect to states and control actions is independent of $\theta^k$, while the gradient terms depend linearly on $\theta^k$. Such a setting is common in goal-reaching tasks where the intent parameter directly represents the goal position.
\begin{proposition}[\textbf{N-PACE v.s treating the peer as an expert}]
\label{propos}
Suppose the stage costs $g_t^k(s_t,a_t^i,a_t^j;\theta^k)$ in \eqref{eq:J} have Hessians with respect to $(s,a^i,a^j)$ that are independent of $\theta^k$, with $\theta^k$ entering only linearly in the first-order terms, and on any compact operating set $s_t \in \Omega$, the policy of each agent admits the form $\pi_t^k(s_t;\theta^k,\hat\theta_t^{-k})=M_t^k(s_t,\hat\theta_t^{-k})\,\theta^k+S_t^k(s_t,\hat\theta_t^{-k})$ with $M_t^k,S_t^k$ bounded. Also, suppose each agent updates its peer parameter online by the gradient descent learning rule \eqref{eq:learning_dy}, then the convergence of intent estimation in N-PACE is guaranteed for sufficiently small learning rate $\alpha$, whereas methods that treat the peer as an expert/complete-information agent can fail to converge for any value of learning rate. 
\end{proposition}

\begin{proof}
See Appendix~\ref{app:proof prop}.
\end{proof}

\begin{algorithm}[t]
\caption{N-PACE for agent \( i \)}
\label{alg:N-PACE}
\begin{algorithmic}
  \State \textbf{Initialization:} Initialize \( \hat{\theta}^j_0 \) and \( \hat{\theta}^i_0 \).
  \vspace{0.3em}
  \For{each time step $t < T$}
    \State \textbf{1. Policy Generation:}
    \vspace{0.3em}
    \State \hspace{0.1cm} \textbf{a.} Using (\ref{eq:policy}), generate the policy as:
    \State \hspace{0.1cm} \( \pi^i_t(s_t;\theta^i, \hat{\theta}^j_t)=iLQSolver(s_t,\theta^i,\hat{\theta}^j_t) \) 
    \State \hspace{0.1cm} 
    \vspace{0.3em} \textbf{b.} Apply the control \( a^i_t = \pi^i_t(s_t;\theta^i, \hat{\theta}^j_t) \) 
    \State \textbf{2. Prediction:}
    \vspace{0.3em}
    \State \hspace{0.1cm} \textbf{a.} Predict the \( j \)'s policy as
     \( \pi^j_t(s_t;\hat{\theta}^j_t, \hat{\theta}^i_t) \) and model
    \State \hspace{0.1cm} how \( j \)'s predict your  policy \( \pi^i_t(s_t;\hat{\theta}^i_t, \hat{\theta}^j_t) \) using the
    \State \hspace{0.1cm} centralized solver \( iLQSolver(s_t,\hat{\theta}^i_t,\hat{\theta}^j_t) \)
    \vspace{0.3em}
    \State \hspace{0.1cm} \textbf{b.} compute the gradient of each predicted policy 
    \State \hspace{0.1cm} with respect to their corresponding parameter \( \nabla_{\hat{\theta}^k_t}{\pi^k_t} \)
    \State \textbf{3. Learning:}
      \State \hspace{0.1cm} \textbf{a.} observe \( s_t, a^j_t \), form the prediction errors.
      \State \hspace{0.1cm} \textbf{b.} Using the prediction and computed gradients apply
      \State \hspace{0.1cm} the learning dynamics as:  
      \( \hat{\theta}^{i}_{t+1} = h^i_l(\hat{\theta}^{i}_t, \hat{\theta}^{j}_t, a^{i}_t, s_t) \)
      \State \hspace{0.1cm} \( \hat{\theta}^{j}_{t+1} = h^j_l(\hat{\theta}^{j}_t, \hat{\theta}^{i}_t, a^{j}_t, s_t) \)
  \EndFor
  \vspace{0.3em}
\end{algorithmic}
\end{algorithm}

\subsection{N-PACE for Intent Communication}
\label{sec: intent}
The N-PACE algorithm also enables intent communication. 
Notice that the learning dynamics of agent \( -k \), denoted as \( h^{-k}_l \), are fundamentally conditioned on the actions \( a^k_t \) taken by agent \( k \). As a result, each agent can modify its actions to influence the learning process of the other.  
To enable this, we propose to introduce an additional intent communication term in the form of \( \eta(\hat{\theta}^k_t - \theta^k)^2 \) to the instantaneous cost function of each agent, \( g^k_\tau(s_\tau, a^i_\tau, a^j_\tau; \theta^k) \), where \( \eta \) is the weighting scalar.  
This leads to the following optimization problem for each agent \( k \):
\begin{equation}
\begin{aligned}
    \min_{a^k_t}\ & J^{*k}_t(s_t, a^k_t, a^{-k}_t)
    = \sum_{\tau=t}^{T} \Big[\, g^k_\tau(s_\tau, a^k_\tau, a^{-k}_\tau; \theta^k) + \eta\,\|\hat{\theta}^k_{\tau} - \theta^k\|^2 \,\Big] \\
    \text{s.t.}\ & s_{\tau+1} = f_t(s_{\tau}, a^k_{\tau}, a^{-k}_{\tau}),\ \ \hat{\theta}^{-k}_{{\tau}+1} = h^{k}_l(\hat{\theta}^{-k}_{\tau}, \hat{\theta}^{k}_{\tau}, a^{-k}_{\tau}, s_{\tau}),\\
    & \hat{\theta}^{k}_{{\tau}+1} = h^{-k}_l(\hat{\theta}^{k}_{\tau}, \hat{\theta}^{-k}_{\tau}, a^{k}_{\tau}, s_{\tau}).
\end{aligned}
\label{eq:opt_problem}
\end{equation}

The influence problem above is computationally expensive and may be intractable for real-time implementation. Remember that the learning dynamics $h^k_l, h^{-k}_l$ themselves require solving a general-sum game, using ILQgames in this case, for their prediction phase and computing the gradient of that prediction.  
To address this challenge, we propose an approximation method for solving the optimization problem in (\ref{eq:opt_problem}). Instead of considering the evolution of the other agent’s parameter learning throughout the entire horizon $[t,T]$ in (\ref{eq:opt_problem}) to form the cumulative cost $J^{*k}$, we focus only on a single-step update of the other agent’s learning as $\hat{\theta}^{-k}_{{t}+1} = h^{k}_l(\hat{\theta}^{-k}_{t}, \hat{\theta}^{k}_{t}, a^{-k}_{t}, s_{t})$, so we can modify the cumulative cost of agent $k$ at each step as
\begin{equation}
    J^{*k}(s_t,\hat{\theta}^{-k}_{t},a^i_t,a^j_t) = J^{k}(s_t,a^i_t,a^j_t) + \eta \|\hat{\theta}^k_{t+1} - \theta^k\|^2,
    \label{eq:signal}
\end{equation}
where $\hat{\theta}^k_{t+1}$ depends on $a^k_t$ through the learning dynamics $h^{-k}_l$.

\begin{figure*}[t]
\centering
\includegraphics[width=\textwidth]{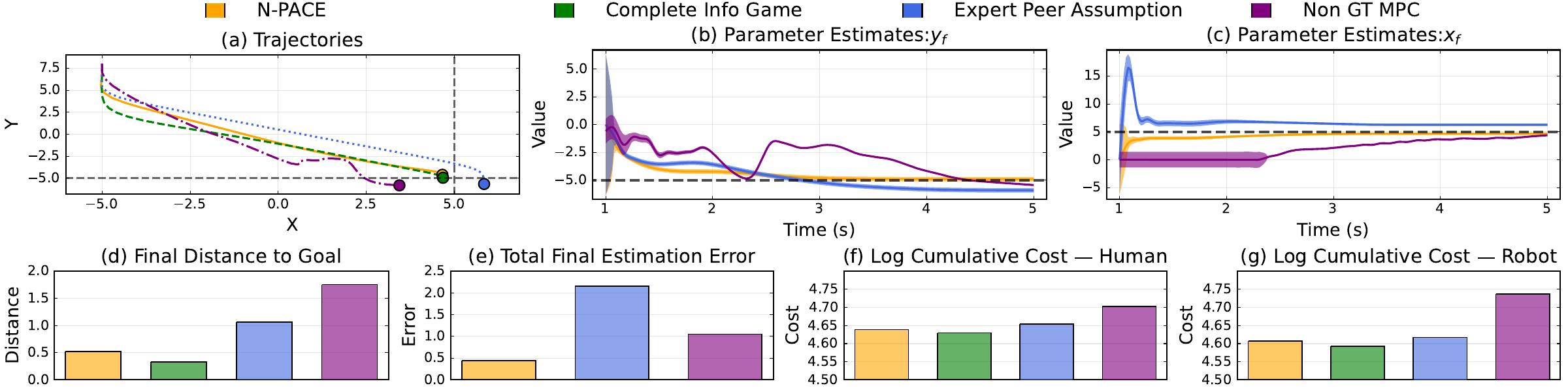}
\caption{Results of Case Study 1, the assistive lunar lander. (a) Trajectory comparison between baselines and N-PACE. (b) and (c) show each agent’s learning performance under different algorithms (d) comparison of final landing position error. (e) Final sum of intent learning error. (f) and (g) cumulative cost comparisons under different approaches for each agent. 
}
\label{fig:exp1_traj}
\end{figure*}

\begin{lemma}[\textbf{Monotone reduction of intent estimation error under intent communication}]\label{lem:monotone-teaching}
Assuming a fixed $t$, ego agent $k$, state $s_t$, peer action $a_t^{-k}$, current beliefs $(\hat{\theta}^k_t,\hat{\theta}^{-k}_t)$, an admissible action domain $\mathcal{A}_t^k\!\subseteq\!\mathbb{R}^{m_k}$, a base objective $C_t^k(a^k)\!=\!J^k(s_t,a^k,a_t^{-k})$, and any learning dynamics $h^{-k}_\ell$. For any $0\le \eta_1<\eta_2$, let action $a_t^{k,\eta}$ be any solution of the optimization problem \emph{(\ref{eq:signal})} at each step $t$,the parameter learning error
$
E_t^k\!\big(a_t^{k,\eta}\big)=\big\|\hat{\theta}^k_{t+1}(a_t^{k,\eta}) - \theta^k\big\|^2
$
satisfies
$
E_t^k\!\big(a_t^{k,\eta_2}\big)\ \le\ E_t^k\!\big(a_t^{k,\eta_1}\big).
$
In particular, $\eta\mapsto E_t^k(a_t^{k,\eta})$ is non-increasing.
\end{lemma}

\begin{proof}
See Appendix~\ref{app:proof:monotone-teaching}.
\end{proof}

\begin{corollary}[\textbf{Intent communication never hurts}]\label{cor:never-hurts}
Let $a_t^{k,0}\in\arg\min_{\mathcal{A}_t^k} C_t^k$ be any no-teaching optimizer and, for $\eta\ge 0$, let $a_t^{k,\eta}\in\arg\min_{\mathcal{A}_t^k}\Phi_t^{k,\eta}$ with $\Phi_t^{k,\eta}:=C_t^k+\eta E_t^k$. Then for every $\eta>0$,
$
E_t^k\!\big(a_t^{k,\eta}\big)\ \le\ E_t^k\!\big(a_t^{k,0}\big).
$
Moreover, if $\arg\min_{\mathcal{A}_t^k} C_t^k \,\cap\, \arg\min_{\mathcal{A}_t^k} E_t^k =\varnothing$ (i,e. disjoint argmin sets), then there exists $\eta^\star>0$ such that for all $\eta>\eta^\star$,
$
E_t^k\!\big(a_t^{k,\eta}\big)\ <\ E_t^k\!\big(a_t^{k,0}\big).
$
\end{corollary}

\begin{proof}
The non-strict inequality is Lemma~\ref{lem:monotone-teaching} with $\eta_1=0$ and $\eta_2=\eta$. For strictness, Define
\[
C_\star:=\min_{\mathcal{A}_t^k} C_t^k,\quad
E_0:=\min_{a\in \arg\min C_t^k} E_t^k(a),\quad
E_\star:=\min_{\mathcal{A}_t^k} E_t^k(a).
\]
Disjointness implies $E_\star<E_0$. Pick $a_E\in\arg\min E_t^k$ and set $\Delta C:=C_t^k(a_E)-C_\star\ge 0$, $\Delta E:=E_0-E_\star>0$. Then
\[
\Phi_t^{k,\eta}(a_E)-\Phi_t^{k,\eta}(a_t^{k,0})=\Delta C-\eta\,\Delta E.
\]
Choosing $\eta^\star:=\Delta C/\Delta E$ makes this negative for all $\eta>\eta^\star$, hence any minimizer $a_t^{k,\eta}$ satisfies
$E_t^k(a_t^{k,\eta})\le E_t^k(a_E)=E_\star<E_0=E_t^k(a_t^{k,0})$.
\end{proof}

 We will show the effectiveness of this approximation for intent communication and empirical evidence validating Lemma \ref{lem:monotone-teaching} and Corollary \ref{cor:never-hurts} in Case Study 3 in Sec. \ref{case}.


\section{Case Studies}
\label{case}
In this section, we show that explicitly modeling a peer’s learning dynamics with N-PACE improves task completion (Case Study 1) and safety (Case Studies 2–3) compare to baseline methods. We also evaluate the robustness of N-PACE to misspecified initial beliefs about the peer’s parameter learning (Case Study 2), and present both an ablation on initialization error and the effect of intent communication (Case Study 3).
Across all case studies, the learning updates use gradient-descent-based–based Bayesian inference with a Gaussian belief. (Simulation Codes are available \href{https://github.com/YousefSoltanian/AAMAS2026-NPACE}{Github}.
 mathematical details of baselines are available in the Appendix)

\subsection{Case Study 1: Assistive Lunar Lander}
Our first example focuses on the importance of N-PACE in task completion. In this example, we study a shared control problem, an assistive lunar lander scenario in which a human agent controls only the angle of the lunar lander $\phi$ by applying the torque $T$, while an autonomous agent controls the thrust force $F$ along the direction of the lunar lander’s angle (assuming no gravity). 
Both agents cooperate to land the lunar lander at the final position \((x_f, y_f)\), but with asymmetric information: the human knows only the final desired \( x \)-position \( x_f \), while the autonomous agent knows only the final desired \( y \)-position \( y_f \). As a result, the human must learn the parameter \( \hat{y}_f \), and the autonomous agent must learn \( \hat{x}_f \) based on observed interactions. The underlying nonlinear dynamics follow $\dot{x} = v_x$, $\dot{y} = v_y$, $\dot{\phi} = \omega$, $\dot{v}_x = F \sin{\phi}$, $\dot{\omega} = T$  the running costs for the human $g_h$ and the agent $g_a$ are defined as $g_a = 10 (x - \hat{x}_f)^2 + 10 (y- y_f)^2 + 10 (\phi - \frac{\pi}{2})^2 + 10 (v_x^2 + v_y^2 + \omega^2) + F^2$ and $g_h = 10 (x - x_f)^2 + 10 (y - \hat{y}_f)^2 + 10 (\phi - \frac{\pi}{2})^2+10(v_x^2 + v_y^2 + \omega^2) + T^2.$
The game has a time horizon of \( T = 5s \) with a sampling time of \( \Delta T = 0.1s \) that is used to discretize the above dynamics.  
In this example, we set \( x_f = 5 \) and \( y_f = -5 \), while the lunar lander starts at the initial position \( (-5,8) \) with a zero initial angle.
We conducted three simulations for this case study:  
1) a complete information game where both agents have full knowledge of \( x_f \) and \( y_f \) as a ground truth baseline, 
2) a game theoretic scenario where each agent treats its peer as an expert, assuming the other agent has complete information, 3)a non-game-theoretic model predictive control (MPC) inspired by the baselines in \cite{peters2024contingency} where each agent uses a history of observed data and fits a linear policy to its partner policy parameters (including the unknown goal position) and using this policy perform a receding horizon MPC, and finally
4) a scenario in which agents operate according to N-PACE.

\begin{figure*}[t]
\centering
\includegraphics[width=\textwidth]{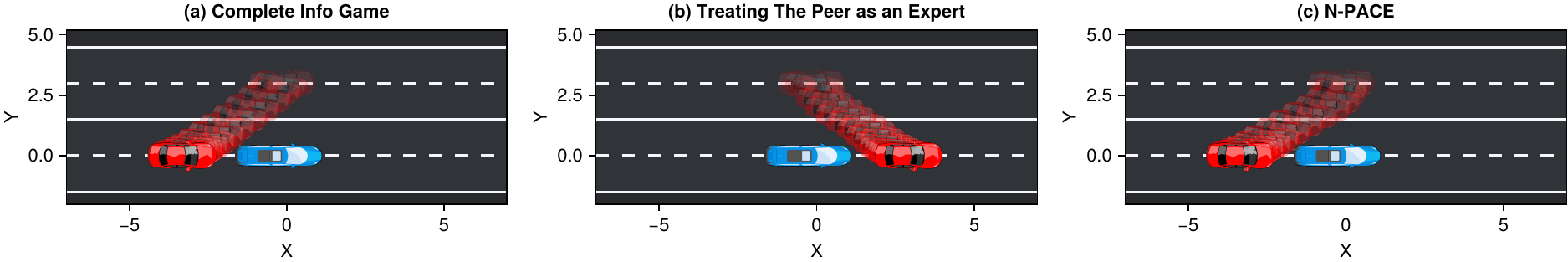}
\caption{Trajectories of the lane-merging case study, the trajectories are plotted relative to the blue car coordinate (a) Safe lane merging under complete information (b) baseline methods, a risky behavior in merging from the aggressive driver (red). (c)Non-risky lane merging in N-PACE similar to the complete game, despite a 30\% error in modeling the peer's initial estimate.}
\label{fig:exp2}
\end{figure*}

\begin{figure}[t]
\centering
\includegraphics[width=\columnwidth]{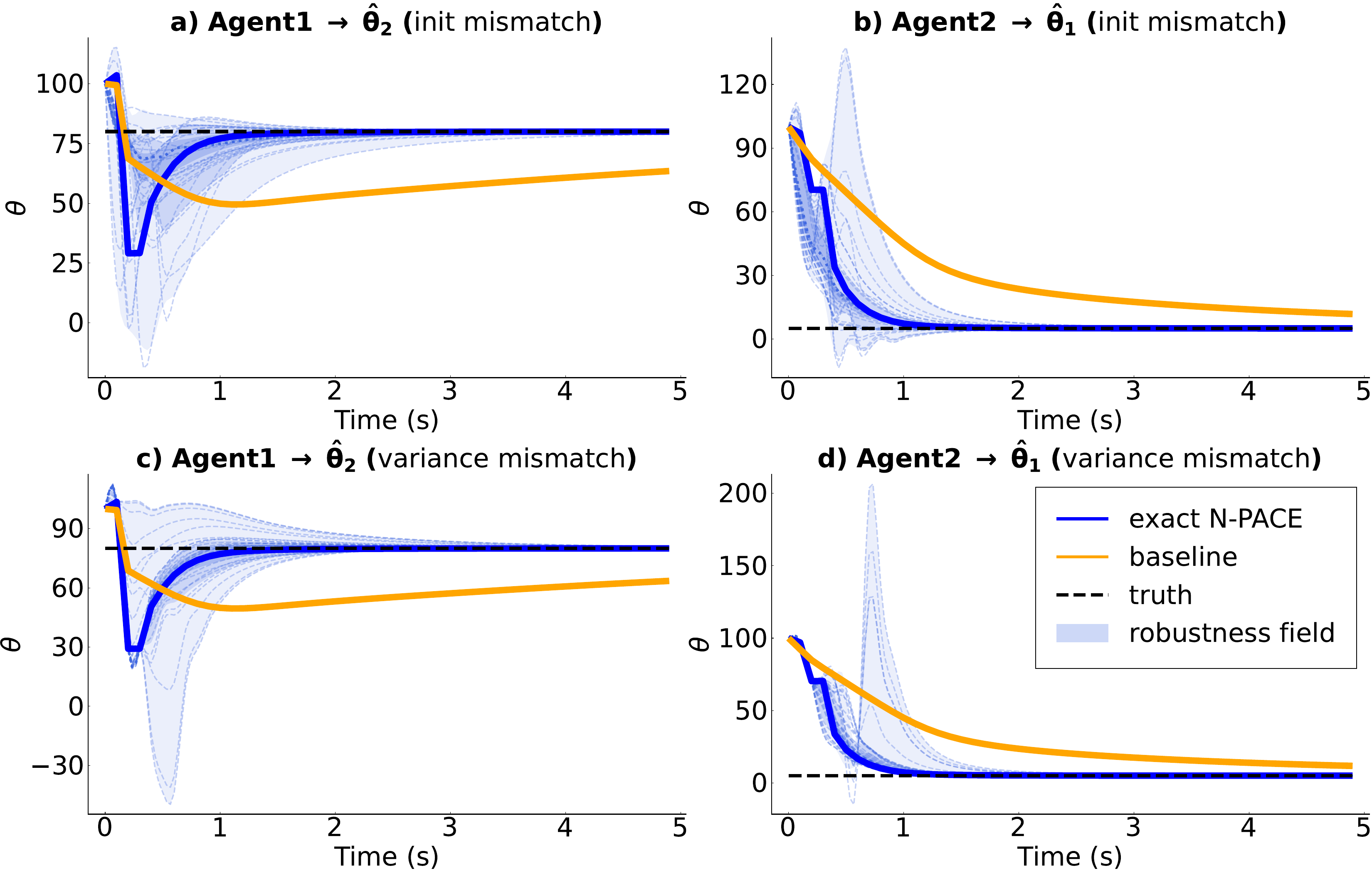}
\caption{a) \& b) Robustness of N-PACE in intent inference under mismatched assumptions about the peer’s initial belief. c) \& d) robustness to wrong choice of initial variance in peer's GD-based Bayesian inference. Envelopes representing 50 mismatched runs.
}
\label{fig:exp2_est}
\end{figure}

Figure \ref{fig:exp1_traj} shows a summary of the results, where Fig. \ref{fig:exp1_traj} (a) presents the trajectory of the lunar lander under different algorithms, explicitly demonstrating the superiority of N-PACE compared to other approaches.  
It can be seen that the final position of the lunar lander in N-PACE is closer to the desired final position, and the entire trajectory aligns more closely with the complete information game trajectory (the local Nash equilibrium). In contrast, treating the other agent as an expert or the MPC-based baseline results in the lunar lander landing at a different position, emphasizing the importance of modeling the learning dynamics of the other agent for successful task completion. Figs. \ref{fig:exp1_traj} (b) and (c) present the results of parameter learning performance for each agent under different algorithms, illustrating how N-PACE fast and successfully converges to the true final position goals in comparison to other baselines.  
Notably, treating the other agent as an expert can lead to a significant overshoot in the learning process. Intuitively, this may be due to a high initial bias in assuming that the peer agent has complete information. \ref{fig:exp1_traj} (f) and (g) interestingly show how N-PACE achieves a lower cumulative cost for each agent compared to other methods in this case study.

\subsection{Case Study 2: Lane Merging}
\label{robust}

This example aims to study the N-PACE robustness to imperfect modeling of the peers' learning dynamics. In our second case study, we consider a lane-merging task between two autonomous vehicles. The first car (blue) moves in a straight line, while the second car (red) attempts to merge into its lane. The blue car can only accelerate or decelerate $a_1$, whereas the red car can control both its acceleration along its direction of motion $a_2$  and its steering angle velocity $\delta$. The game has a 5s horizon with a sampling time of $\Delta{T}=0.1s$ for dynamic discretization. During the game, Each agent is unaware of the aggressiveness level of others in their cost function and must learn it in real-time.
The underlying dynamic is defined as $\dot{x}_1 = v_1$, $\dot{v}_1 = a_1$, $\dot{x}_2 = v_2 \cos{\phi_2}$, $\dot{y}_2 = v_2 \sin{\phi_2}$, $\dot{\phi}_2 = \delta_2,$ $\dot{v}_2 = a_2$, and the running cost $g$ for each car is $g_{blue} = 0.1 (x_1 - 25)^2 + 0.1 v_1^2+ \theta_1 e^{-( d(t) - d_{\text{safe}} )^2} + a_1^2,$ and $g_{red} = 1 y_2^2 + 10 \phi_2^2 + 0.1 v_2^2+ 0.1 (x_2 - 25)^2
+ \theta_2 e^{-( d(t) - d_{\text{safe}} )^2} + a_2^2 + \delta_2^2$.
where \( d(t)=(x_1 - x_2)^2 + y_2^2 \) is the distance between the two cars, \( d_{\text{safe}} \) is the safe distance they aim to maintain, and \( \theta_k \) is a parameter in the range \([0,100]\) that represents the aggressiveness of each agent (lower \( \theta_k \) corresponds to more aggressive behavior). 
Each agent initially assumes that the other agent is a non-aggressive, safety-conscious driver with \( \hat{\theta}^k_0 = 100 \). However, both agents exhibit some level of aggressiveness: the first agent (the blue car) is highly aggressive with \( \theta_1 = 5 \), while the second agent (the red car) has a moderate aggressiveness level of \( \theta_2 = 80 \). 

To test robustness to misspecified initial beliefs, we introduced an initialization error $e_{init} = \pm70\%$ in each agent’s guess of the other’s initial estimation and randomly sampled $50$ "initial parameter estimate of others" for each agent from this error bound. Equivalently, each agent $k$ for each run uniformly samples from the range $\hat{\theta}^k \in [30,170]$ and sets it as the initial estimate of agent $-k$ (although the true initial estimates are \( \hat{\theta}^k_0 = 100 \)).
Across all runs, as shown in figure \ref{fig:exp2_est} a) and b), N-PACE converged to the peer’s true intent and outperformed the baseline, indicating robustness to initialization mismatch. For higher error values, we observed N-PACE starting to show diverging behavior. Also, to better represent the effect of errors in the modeling of peers' learning dynamics, we used an error $e_{var} = \pm55\%$ in how each agent models the initial variance in gradient-based Bayesian inference of the other agent as a way to model misspecified convergence speeds of the modeled learning dynamics. We randomly sampled $50$ variances from this range and ran our simulations. Again, across all cases as shown in figure \ref{fig:exp2_est} c) and d) we observed N-PACE converging and outperforming the baseline.
Indicating how N-PACE can lead to better intent inference despite the mismatch in modeling the initial estimate or learning dynamics of others. 
Qualitative trajectory comparison results and intent inference Results appear in Fig.~\ref{fig:exp2} and Fig.~\ref{fig:exp2_est}  

\begin{figure}[t]
\centering
\includegraphics[width=\columnwidth]{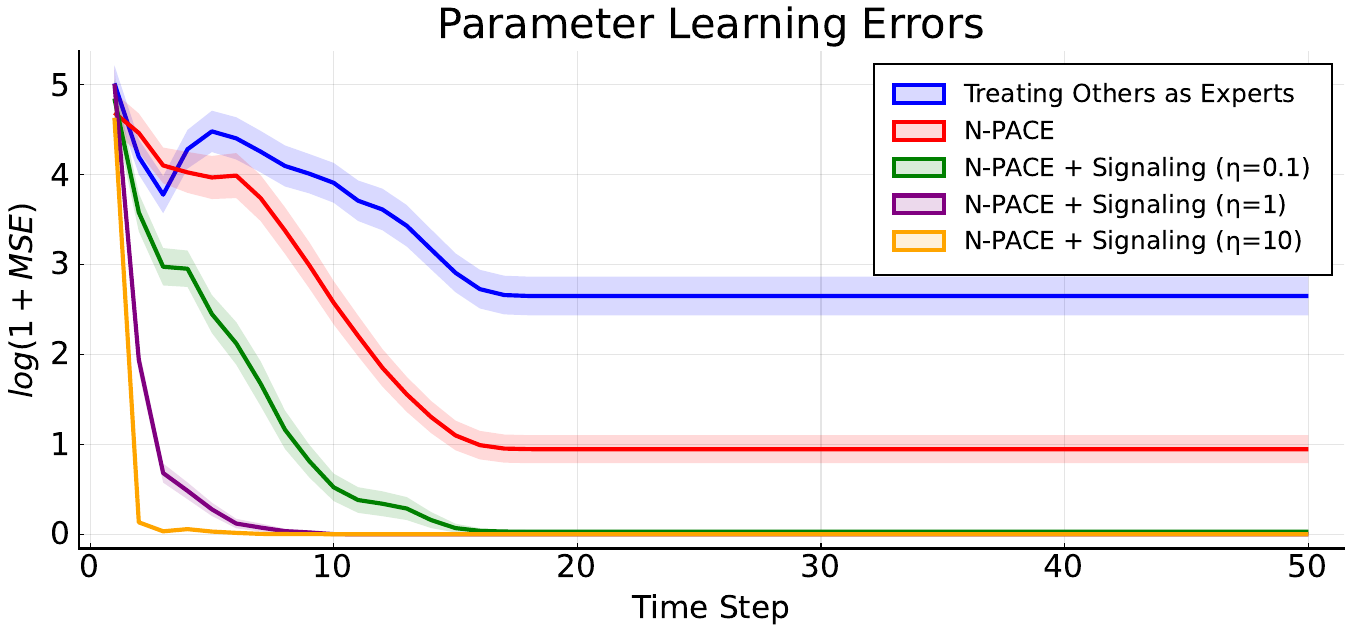}
\caption{ Comparison of \(\log(1+\text{MSE})\) learning error across 300 Monte Carlo simulations of case study 3 with random values of aggressiveness.
}
\label{fig:exp3_est}
\end{figure}
\vspace{-4pt}

\subsection{Case Study 3: Interaction Driving}
\begin{figure*}[t]
\centering
\includegraphics[width=\textwidth]{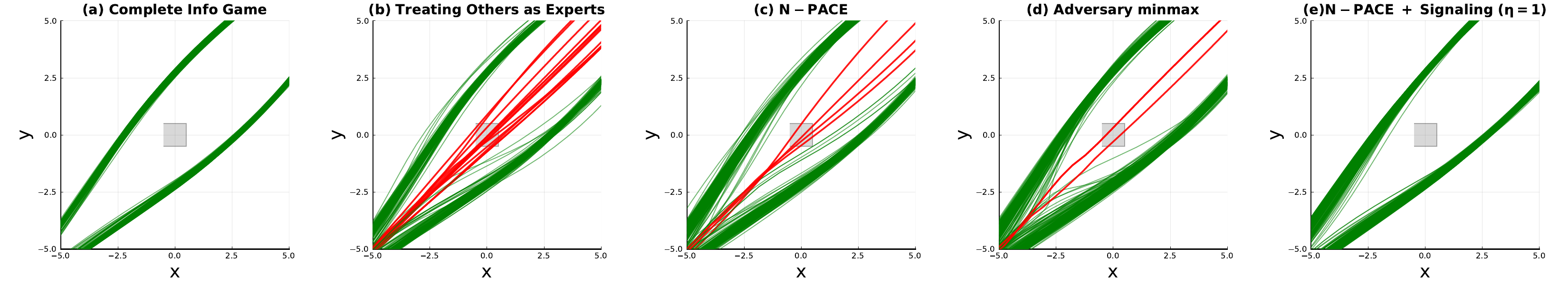}
\caption{Trajectory \((x,y)\) results of 300 simulations of the two-player incomplete information intersection game with random values of aggressiveness. Gray boxes indicate the collision area. (a) The complete information game results in no collisions. (b) The expert peer assumption method results in 13 collisions. (c) N-PACE results in only 4 collisions. (d) A conservative adversarial baseline still results in 2 collision cases. (e) N-PACE, combined with intent communication, achieves 0 collisions and the closest trajectories to those in the complete information game.
}
\label{fig:exp3_traj}
\end{figure*}

In our third case study, to extensively study N-PACE and the intent communication approach proposed in (\ref{eq:signal}), we conduct Monte Carlo simulations of an intersection driving game between two agents under incomplete information, similar to the scenario shown in Fig. \ref{fig1}. In this case study, two autonomous vehicles attempt to navigate through an intersection while being unaware of each other's aggressiveness parameter, \(\theta_k\). The agents are only able to control their acceleration along with the direction of their movement $x,y$. The dynamic of the game follows kinematic equations $\dot{x}_1 = v_1$, $\dot{v}_1 = a_1$, $\dot{y}_2 = v_2$, $\dot{v}_2 = a_2$, and each agent's cost function is formed as $g_1 = (x_1 - 8)^2 + v_1^2 
\;+\; 10\,\theta_1\,\exp\bigl(-\bigl(d(t) -d_{\mathrm{safe}}\bigr)\bigr)\;+\; a_1^2$ and $g_2= (y_2 - 8)^2 + v_2^2 
\;+\; 10\,\theta_2\,\exp\bigl(-\bigl(d(t) - d_{\mathrm{safe}}\bigr)\bigr)\;+\; a_2^2.$
where $d(t) = x_1^2 +y_2^2$ is the distance between two vehicles and $d_{safe}$ is the safe distance they try to maintain.
The aggressiveness of each agent is also defined same as Case Study 2. The game continues for $5$ sec with a sampling time of $\Delta{T} = 0.1$ sec, resulting in 50 samples.

In this example, each agent's aggressiveness parameter, \(\theta_k\), takes values in the range $\theta_k \in [20,50]$, where a lower \(\theta_k\) indicates a more aggressive agent. In our Monte Carlo study, we sampled 300 aggressiveness values uniformly from \([20,50]\) for each agent. Using these 300 sampled values, we ran the game 300 times under five different scenarios:  
1) Complete information,  
2) The expert peer agent assumption,
3) An adversary min-max game baseline where each agent $k$, instead of trying to do intent inference, at each step decides conservatively based on $min_{a^k}max_{\theta^{-k} \in [20,50]}J^k(s_t,a^i_t,a^j_t)$
4) N-PACE,  
5) N-PACE with intent communication, as proposed in Sec. \ref{sec: intent}, for different intent communication weighting values \(\eta = 0.1, 1, 10\).
The benchmarking results of these studies are shown in Table (\ref{tab:cs3-benchmark}). As expected, under the complete information game scenario, no collisions were observed.
The peer expert assumption resulted in 13 collisions. In contrast, N-PACE resulted in only 4 collisions. More interestingly, although the minmax approach was a conservative approach with a high average intersection crossing time, it still resulted in 2 collision cases, where on the other hand our intent communication approach (built on top of N-PACE) not only ensured safe crossing for all 300 scenarios, it also achieved the fastest crossing time among all incomplete-information scenarios, ensuring both safety and smooth interaction. This was also the case for \(\eta = 10\), while only a single collision occurred for \(\eta = 0.1\). The trajectories from these 300 simulations for different scenarios are shown in Fig. \ref{fig:exp3_traj}, where we plot the \(x\)-position of the first agent against the \(y\)-position of the second agent.
These results highlight the improvements in both inference accuracy and safety assurance using the proposed intent communication framework built on top of N-PACE, although at the expense of increasing the average control effort (see table \ref{tab:cs3-benchmark}).
Finally, Fig. \ref{fig:exp3_est} presents the log of the Mean Squared learning Error (MSE) of both agents across all 300 Monte Carlo simulations. We use \(\log(1 + \text{MSE})\) for better visualization and comparison.  Interestingly, the results in Fig. \ref{fig:exp3_est} align well with the theoretical results in corollary \ref{cor:never-hurts} and lemma \ref{lem:monotone-teaching} that state increasing $\eta$ always reduces the estimation error. 
These empirical results highlight how intentional communication improves learning performance, particularly when the signaling is intensified using higher \(\eta\).

\begin{table}[h]
\centering
\small
\setlength{\tabcolsep}{2pt}
\renewcommand{\arraystretch}{0.8}
\caption{Case study 3 benchmark}
\label{tab:cs3-benchmark}
\begin{tabular*}{\columnwidth}{@{\extracolsep{\fill}} lccc @{}}
\toprule
\thead{Method\\} &
\thead{Failure rate\\} &
\thead{Average control\\effort [$m/s^2$]}  &
\thead{Average time\\to cross [s]} \\
\midrule
Complete game            & \textbf{0.0\%} & 5.00 & 1.497 $\pm$ 0.043 \\
Expert peer assumption   & 4.33\%  & 4.92       & $1.593 \pm 0.044$ \\
N\text{-}PACE            & 1.33\% & 4.86      & $1.599 \pm 0.054$ \\
N\text{-}PACE + $\eta=0.1$  & 0.33\% & 5.42 & 1.521 $\pm$ 0.009 \\
N\text{-}PACE + $\eta=1.0$  & \textbf{0.0\%} & 5.44 & 1.499 $\pm$ 0.029 \\
N\text{-}PACE + $\eta=10.0$ & \textbf{0.0\%} & 5.45 & \textbf{1.489 $\pm$ 0.035} \\
Minmax                   & 0.66\% & 6.23         & $1.651 \pm 0.165$ \\
\bottomrule
\end{tabular*}
\vspace{-12pt}
\end{table}

\section{CONCLUSIONS AND FUTURE WORK}

In this paper, we presented a Nonlinear Peer-Aware Cost Estimation (N-PACE) framework for solving general-sum dynamic games with continuous action-states and continuous intent space, which relies on treating the peer agent as a learning agent while inferring its cost parameters. We demonstrated the importance of accounting for the learning behavior of others and showed how this framework outperformed the approach of treating the peer agent as a complete-information. Additionally, we proposed an intent communication framework built on top of N-PACE for signaling in incomplete information games. 
N-PACE assumes access to the learning dynamics of the peer agent. While such information is often available in many multi-agent systems, there are scenarios—such as human-robot interactions—where the learning dynamics may be unknown. The question of "learning the learning dynamics" presents an interesting research direction for future work, particularly in light of recent studies exploring the use of transformers for this purpose \cite{tian2023towards,kedia2024interact}. Also,  the mentioned works or cognitive research finding stating that humans can be modeled as Bayesian learners \cite{griffiths2010probabilistic}, pave the way for future applications of N-PACE in HRI.
\begin{acks}
This material is based upon work supported by the National Science Foundation under Grant No. 1944833. Any opinions, findings, and conclusions or recommendations expressed in this material are those of the author(s) and do not necessarily reflect the views of the National Science Foundation.
\end{acks}
\bibliographystyle{ACM-Reference-Format} 
\bibliography{sample}
\vspace{-2pt}
\appendix            
\label{appendix}
\section{Proofs}
\begin{proof}[Proof of Proposition~\ref{propos}]
\label{app:proof prop}
Let $e_t^k:=\hat\theta_t^k-\theta^k$ be the parameter estimation error of agent $k$ using the N-PACE method, and $\tilde e_t^k:=\tilde\theta_t^k-\theta^k$ be the parameter estimation error under the assumption of having an expert peer, for $k\in\{i,j\}$. The realized peer actions obey $a_t^k=M_t^k(s_t,\hat\theta_t^{-k})\,\theta^k+S_t^k(s_t,\hat\theta_t^{-k})+w_t^k$ (linear in $\theta^k$) where $w_t^k$ is a bounded variance measurement noise. Then the N-PACE and Expert peer assumption method error recursions are:
\[
\begin{aligned}
\text{(N-PACE)}\quad
e_{t+1}^i&=\big(I-\alpha\,M_t^{i\top}M_t^i\big)e_t^i+\alpha\,M_t^{i\top}w_t^i,\\
e_{t+1}^j&=\big(I-\alpha\,M_t^{j\top}M_t^j\big)e_t^j+\alpha\,M_t^{j\top}w_t^j,\\[-1mm]
\text{(Expert assu)} \quad
\tilde e_{t+1}^i&=\big(I-\alpha\,\tilde M_t^{i\top}\tilde M_t^i\big)\tilde e_t^i-\alpha\,\tilde M_t^{i\top}C_t^i\,\tilde e_t^j+\alpha\,\tilde M_t^{i\top}w_t^i,\\
\tilde e_{t+1}^j&=\big(I-\alpha\,\tilde M_t^{j\top}\tilde M_t^j\big)\tilde e_t^j-\alpha\,\tilde M_t^{j\top}C_t^j\,\tilde e_t^i+\alpha\,\tilde M_t^{j\top}w_t^j,
\end{aligned}
\]
where $\tilde M_t^k:=M_t^k(s_t,\theta^{-k})$, and the off-diagonal terms $C_t^i,C_t^j$ \emph{include both $M$- and $S$-sensitivities} via the mean-value form
\[
\begin{aligned}
C_t^i&=\int_0^1\! \frac{\partial}{\xi_t^i}\Big( M_t^i\big(s_t,\xi_t^i(\tau)\big)\,\theta^i+ S_t^i\big(s_t,\xi_t^i(\tau)\big)\Big)\,d\tau,\\
&\xi_t^i(\tau):=\hat\theta_t^j+\tau(\theta^j-\hat\theta_t^j)
\end{aligned}
\]
and analogously for $C_t^j$; thus $\big(M_t^i(s_t,\theta^j)-M_t^i(s_t,\hat\theta_t^j)\big)\theta^i+\big(S_t^i(s_t,\theta^j)-S_t^i(s_t,\hat\theta_t^j)\big)=C_t^i\,(\theta^j-\hat\theta_t^j)=-C_t^i\,\tilde e_t^j$ (and symmetrically for $j$). If $\|M_t^k\|\le L$, $\mathbb E\|w_t^k\|^2\le\sigma_k^2$, and there is persistent excitation $\lambda_{\min}\big(\mathbb E[M_t^{k\top}M_t^k]\big)\ge\mu_k>0$, then from N-PACE errors dynamics we have $\mathbb E[\|e_{t+1}^k\|^2\mid\mathcal F_t]\le(1-2\alpha\mu_k+\alpha^2L^4)\|e_t^k\|^2+\alpha^2L^2\sigma_k^2$, so any $0<\alpha<2\mu/L^4$ gives contraction in mean square (a.s.\ convergence with diminishing stepsizes; bounded MSE floor with constant stepsizes). In contrast, expert assymption error dynamics contains the off-diagonal couplings $-\alpha\,\tilde M_t^{i\top}C_t^i$ and $-\alpha\,\tilde M_t^{j\top}C_t^j$ that arise from \emph{both} the $M$- and $S$-term mismatches; unless these are uniformly dominated by the diagonal self-excitation (e.g., the symmetric part of $\tilde M_t^{k\top}\tilde M_t^k$ exceeds $\tilde M_t^{k\top}C_t^k$ for $k\in\{i,j\}$), the one-step map need not be a contraction for any $\alpha>0$, so the intent inference while assuming the peer agent is an expert is not guaranteed to converge and may even result in unstable parameter estimation/inference.
\end{proof}


\begin{proof}[Proof of Lemma~\ref{lem:monotone-teaching}]\label{app:proof:monotone-teaching}
Fix $t$, agent $k$, and data $(s_t,a_t^{-k},\hat{\theta}^k_t,\hat{\theta}^{-k}_t)$. Let the admissible set $\mathcal{A}_t^k$ be nonempty, and suppose minimizers $a_t^{k,\eta}\in\arg\min_{a^k\in\mathcal{A}_t^k}\Phi_t^{k,\eta}(a^k)$ exist for the weights $\eta$ under consideration (e.g., compact $\mathcal{A}_t^k$ and continuity of $C_t^k,E_t^k$ suffice). Recall the definitions:
\[
\begin{aligned}
C_t^k(a^k) &:= J^k(s_t, a^k, a_t^{-k}), \quad
\hat{\theta}^k_{t+1}(a^k):= h^{-k}_\ell(\hat{\theta}^k_t,\hat{\theta}^{-k}_t, a^k, s_t),\\
E_t^k(a^k) &:= \big\|\hat{\theta}^k_{t+1}(a^k) - \theta^k\big\|^2, \quad
\Phi_t^{k,\eta}(a^k):= C_t^k(a^k) + \eta\,E_t^k(a^k).
\end{aligned}
\]

Take any $0\le \eta_1<\eta_2$. By optimality of $a_t^{k,\eta_2}$ for $\Phi_t^{k,\eta_2}$ and $a_t^{k,\eta_1}$ for $\Phi_t^{k,\eta_1}$,
\begin{align}
C_t^k(a_t^{k,\eta_2})+\eta_2\,E_t^k(a_t^{k,\eta_2})
&\le C_t^k(a_t^{k,\eta_1})+\eta_2\,E_t^k(a_t^{k,\eta_1}), \label{eq:opt1}\\
C_t^k(a_t^{k,\eta_1})+\eta_1\,E_t^k(a_t^{k,\eta_1})
&\le C_t^k(a_t^{k,\eta_2})+\eta_1\,E_t^k(a_t^{k,\eta_2}). \label{eq:opt2}
\end{align}
Add \eqref{eq:opt1} and \eqref{eq:opt2} and cancel the $C_t^k$ terms on opposite sides to obtain
\[
(\eta_2-\eta_1)\,\Big(E_t^k(a_t^{k,\eta_2})-E_t^k(a_t^{k,\eta_1})\Big)\ \le\ 0.
\]
Since $\eta_2-\eta_1>0$, it follows that $E_t^k(a_t^{k,\eta_2})\le E_t^k(a_t^{k,\eta_1})$. As $\eta_1,\eta_2$ were arbitrary with $\eta_1<\eta_2$, the map $\eta\mapsto E_t^k(a_t^{k,\eta})$ is non-increasing.
\end{proof}

\section{Details of Bayesian Gradient Descent Learning Function}
\label{app:bayesian_gd}

In our case study implementations, we considered a Bayesian interpretation of the intent learning rule used in Eq.~(\ref{eq:learning_dy}). 
Agent \(k\) maintains a Gaussian belief over the opponent’s intent parameter,
\(b_t(\theta^{-k}) = \mathcal{N}(\mu_t^{-k},\,\Sigma_t^{-k})\).
At each time step, after observing the opponent’s control \(a_t^{-k}\) and the current state \(s_t\), the belief is updated using a Gaussian observation model:
\[
p(a_t^{-k}\mid s_t,\theta^{-k}) = 
\mathcal{N}\!\Big(a_t^{-k}\,;\,\pi_t^{-k}(s_t;\theta^{-k},\hat{\theta}^k_t),\,R_t\Big),
\]
where \(R_t\succ 0\) is the observation noise covariance, typically \(R_t=\sigma_a^2 I_m\).
This assumption implies that the opponent’s action is a noisy realization of their nominal policy 
\(\pi_t^{-k}(s_t;\theta^{-k},\hat{\theta}^k_t)\).
Applying Bayes’ rule yields
\[
b_{t+1}(\theta^{-k}) \propto p(a_t^{-k}\mid s_t,\theta^{-k})\,b_t(\theta^{-k}).
\]

Under the Gaussian–Gaussian model, the posterior remains Gaussian with analytically tractable updates for the mean and covariance.
Let \(J_t^{-k} = \nabla_{\theta^{-k}} \pi_t^{-k}(s_t;\mu_t^{-k},\hat{\theta}^k_t)\in\mathbb{R}^{m\times p}\) be the policy Jacobian with respect to the opponent’s intent parameters. 
Then, the belief updates are given by
\begingroup\small
\begin{align}
\mu_{t+1}^{-k} &= \mu_t^{-k} + \Sigma_t^{-k} J_t^{-k\top}
                 \Big(R_t + J_t^{-k}\Sigma_t^{-k}J_t^{-k\top}\Big)^{-1}
                 \Big(a_t^{-k} - \pi_t^{-k}(s_t;\mu_t^{-k},\hat{\theta}^k_t)\Big), \label{eq:bgd_mean}\\[3pt]
\Sigma_{t+1}^{-k} &= \Sigma_t^{-k} - \Sigma_t^{-k} J_t^{-k\top}
                 \Big(R_t + J_t^{-k}\Sigma_t^{-k}J_t^{-k\top}\Big)^{-1}
                 J_t^{-k}\Sigma_t^{-k}. \label{eq:bgd_var}
\end{align}
\endgroup

Equation~\eqref{eq:bgd_mean} can be viewed as a Bayesian counterpart to the point-estimate update in Eq.~(\ref{eq:learning_dy}), 
where the gain term \(\Sigma_t^{-k}J_t^{-k\top}(R_t+J_t^{-k}\Sigma_t^{-k}J_t^{-k\top})^{-1}\) acts as an adaptive learning rate that shrinks as uncertainty decreases.
The covariance recursion~\eqref{eq:bgd_var} quantifies this information gain: as the agent collects more consistent observations, the posterior variance \(\Sigma_t^{-k}\) contracts, naturally tempering the step size in subsequent updates.
This Bayesian gradient-descent interpretation thus generalizes the deterministic gradient rule by embedding uncertainty weighting and curvature correction into the learning dynamics.

In Case Study~2, during the robustness analysis, the term \emph{mismatched initial estimates} refers to the initial values of \(\mu_t\), whereas \emph{mismatched variances} refer to the initial values of \(\Sigma_t\). These parameters effectively modulate the convergence rate of the mean estimates and can be interpreted as representing discrepancies in how the agents model each other’s learning dynamics.

\section{Details of The Non Game Theoretic MPC in Case Study 1}
\label{sec:peeraware-mpc}

At time $t$, agent $k\!\in\!\{i,j\}$ represents the other agent $-k$ by a \emph{linear goal–seeking} feedback policy (affine in state)
\begin{equation}
\pi^{\,-k}(s)\;=\;K^{\,-k}\!\big(s - s_{\text{goal}}(\hat\theta^{\, -k})\big)\;+\;b^{\,-k},
\label{eq:peer-linear}
\end{equation}
where $K^{\,-k},b^{\,-k}$ are planning gains, and $s_{\text{goal}}(\hat\theta^{\, -k}_\tau)$ encodes the peer’s (unknown) goal components (In lunar lander case, $\hat x_f$ or $\hat y_f$, and desired finla $\pi/2$ angle and zero velocities). 
The parameters $\hat\theta^{\, -k}$ and $K^{\,-k},b^{\,-k}$ are \emph{updated online from recent $(s,a^{-k})$} via an Extended Kalman Filtering, so the modeled policy better matches the observed peer actions.

Given the current estimates $(\hat\theta^{\, -k}_t,K^{\,-k}_{t},b^{\,-k}_{t})$, agent $k$ optimizes \emph{its own} receding–horizon cost:
\begin{equation}
\begin{aligned}
\min_{\{a^k_\tau\}} \ \ 
&J^k_t \;=\; \sum_{\tau=t}^{T} g^k_\tau\!\Big(s_\tau,\ a^k_\tau,\ \pi^{\,-k}(s_\tau);\ \theta^k\Big)
\; \\
\text{s.t.}\ \ 
&s_{\tau+1} \;=\; f\!\Big(s_\tau,\ a^k_\tau,\ \pi^{\,-k}(s_\tau)\Big),
\end{aligned}
\label{eq:peeraware-mpc}
\end{equation}
So at each step agent $k$ (i) updates $\hat\theta^{\, -k}$ and, $K^{\,-k},b^{\,-k}$ from the latest action observations using EKF; 
(ii) \emph{fixes} the peer model \eqref{eq:peer-linear}; 
(iii) applies the best response to the peer's modeled policy by solve \eqref{eq:peeraware-mpc} with iLQR, and using the first action $a^k_t$.

\section{Details of Min–Max Adversary Controller in Case Study~3}
\label{subsec:minmax}

The \emph{min–max adversary} baseline models the opponent’s intent parameter 
$\theta^{-k}$ as an unknown but bounded disturbance within a compact set 
$\Theta^{-k}$. Rather than inferring $\theta^{-k}$, each agent $k\!\in\!\{i,j\}$ 
adopts a robust strategy that minimizes its own cost under the \emph{worst-case} 
opponent intent.

For any pair $(\theta^i,\theta^j)$, the agents solve a finite-horizon 
complete-information differential game using iLQGames, yielding local Nash control 
sequences parametrized in $(\theta^i,\theta^j)$.
\[
(\mathbf{u}^{i*},\,\mathbf{u}^{j*})
=\mathrm{iLQGame}\!\big(s_t;\theta^i,\theta^j,H\big),
\quad
s_{\tau+1}^*=f_\tau(s_\tau^*,u_\tau^{i*},u_\tau^{j*}),
\]
Consequently, the associated value for agent $k$ is:
\[
\Phi^k(\theta^{-k};s_t)
=\sum_{\tau=t}^{T}
g^k_\tau\!\big(s_\tau^*,u_\tau^{i*},u_\tau^{j*};\theta^k\big).
\]

Agent $k$ then chooses its control based on the opponent's worst possible intent 
$\theta^{-k}_{\mathrm{wc}}$ that maximizes its own cost:
\[
\theta^{-k}_{\mathrm{wc}}
\in
\arg\max_{\theta^{-k}\in\Theta^{-k}}
\Phi^k(\theta^{-k};s_t),
\qquad
a_t^k
=
u_t^{k*}\!\big(s_t;\theta^k,\theta^{-k}_{\mathrm{wc}}\big).
\]
The opponent’s parameter is thus treated as an \emph{adversarial latent variable}, 
and the agent executes the first control from the corresponding ILQ equilibrium.

After applying $(a_t^i,a_t^j)$, the system evolves via 
$s_{t+1}=f_t(s_t,a_t^i,a_t^j)$, and the entire procedure repeats. 
The maximization over $\theta^{-k}$ is performed numerically using a 
derivative-free golden-section search on the bounded interval $\Theta^{-k}$.
This baseline performs no Bayesian update or learning, representing a 
\emph{worst-case robust controller} that trades adaptivity for guaranteed safety.
It needs to be noted that although agent $k$ performs based on the  opponent's worst possible intent $\theta^{-k}_{\mathrm{wc}}$, in the iLQGame that it is solving, the assumption of a complete information peer agent is still hidden. This explains why, even under an adversary-minimax controller that is supposed to be safe, we still observed two collisions in case study 3.

\section{Computational Cost and Scalability}
\label{sec:comp-cost}
All experiments were run in \textbf{Julia 1.11.2} on a \textbf{13th Gen Intel Core i7-1360P} (16 threads). 
Unless noted otherwise, simulations use sampling time $\Delta T = 0.1\,$s and planning horizon $T=50$. 
We report wall-clock time \emph{per control step} for the deployment-relevant agent (agent~2), split into 
\emph{control} (policy solve) and \emph{estimation} (parameter update), with medians and p95 in brackets.

Our control loop mirrors the ILQGames kernel~\cite{fridovich2020efficient}: at each iteration and horizon step we (i) linearize the dynamics and quadraticize costs, then (ii) solve the coupled LQ game via a backward pass that inverts a block of size $m=\sum_i n_{u,i}$ (total control dimension). The dominant per-iteration cost scales as
\[
\mathcal{O}\!\left(H \big(n_x^3 + m^3\big)\right),
\]
multiplied by the number of iLQ iterations; this matches the complexity discussion in~\cite{fridovich2020efficient}. Our \emph{control} times in Table~\ref{tab:comp-cost} (medians 11.6--19.6\,ms; p95 $<$ 32\,ms) are therefore directly comparable to ILQGames timing reports, while our \emph{estimation} adds a modest constant-factor overhead (another predictive solve + gradients).

At sampling $\Delta T{=}0.1$\,s (10\,Hz) of the case studies, all \emph{total} p95s are $\leq 64$\,ms, leaving comfortable slack to the 100\,ms budget. Because the dominant kernel is identical to ILQGames as in ~\cite{fridovich2020efficient}, our applicability as $n_x$ or the number of players grows is similar: increasing $n_x$ mainly affects the $n_x^3$ term, and adding players/control channels increases $m$ and thus the $m^3$ term. In practice, we found that warm-starting and capping ILQ iterations keep medians low and tails tight; for larger $n_x$ or more players, the same ILQGames heuristics (trajectory/state reuse, preallocation, looser tolerances in the \emph{predictive} solve) might preserve 10\,Hz operation with comparable margins.

\begin{table}[t]
\centering
\small
\begin{tabular}{lccc}
\toprule
Task & Control [ms] & Estimation [ms] & Total [ms] \\
\midrule
Lunar lander ($n_x{=}6,n_u{=}2$) & 11.6 [26.6]  & 12.3 [27.2]  & 24.2 [55.6] \\
Lane merging ($n_x{=}6,n_u{=}3$) & 18.5 [31.1] & 21.3 [33.2]  & 39.0 [63.7] \\
Intersection ($n_x{=}4,n_u{=}2$) & 14.8 [25.6] & 17.1 [30.6]  & 33.2 [54.7] \\
\bottomrule
\end{tabular}
\caption{Per-step wall-clock runtime for agent~2 (medians; p95 in brackets). 
All runs: Julia~1.11.2 on Intel i7-1360P (16 threads), $\Delta T{=}0.1$\,s, $T{=}50$.}
\label{tab:comp-cost}
\end{table}

\section{Extension to Multi-Agent Scenarios}
N-PACE can be extended to more agents, similar to how ILQGame\cite{fridovich2020efficient} scales to multi-agent settings: each agent would model the learning of all others, increasing the number of predictive ILQ solves and terms in equation (\ref{eq: PACE})  as:

\begin{equation}
\label{eq:PACE-m}
\begin{aligned}
\hat{\theta}^{1}_{t+1} &= h^{1}_{l}\!\Big(\hat{\theta}^{1}_{t},\,\hat{\theta}^{2}_{t},\,\ldots,\,\hat{\theta}^{m}_{t},\; a^{1}_{t},\,a^{2}_{t},\,\ldots,\,a^{m}_{t},\; s_t\Big),\\
\hat{\theta}^{2}_{t+1} &= h^{2}_{l}\!\Big(\hat{\theta}^{1}_{t},\,\hat{\theta}^{2}_{t},\,\ldots,\,\hat{\theta}^{m}_{t},\; a^{1}_{t},\,a^{2}_{t},\,\ldots,\,a^{m}_{t},\; s_t\Big),\\
&\ \vdots\\
\hat{\theta}^{m}_{t+1} &= h^{m}_{l}\!\Big(\hat{\theta}^{1}_{t},\,\hat{\theta}^{2}_{t},\,\ldots,\,\hat{\theta}^{m}_{t},\; a^{1}_{t},\,a^{2}_{t},\,\ldots,\,a^{m}_{t},\; s_t\Big).
\end{aligned}
\end{equation}

The main limitation is computational cost; as we showed, N-PACE scales as O(m3), comparable to ILQGame\cite{fridovich2020efficient} , making extensions feasible to at least 3–4 agents. To illustrate the multi-agent extension conceptually, we include a 3-vehicle platooning experiment similar to the experiment in \cite{li2024intent} in table \ref{tab:platooning3}.
The reported error convergence results are for the worst among all three 3 vehicles, showing how N-PACE outperforms baselines even in a 3-player interaction.

\begin{table}[h]
\centering
\small
\setlength{\tabcolsep}{2pt}
\renewcommand{\arraystretch}{0.9}
\begin{tabular*}{\columnwidth}{@{\extracolsep{\fill}} lccc @{}}
\toprule
Method & Failure [\%] & $t_{10\%}$ [s] & Final err. [\%] \\
\midrule
Expert baseline & 10.33 & --  & 14 \\
N-PACE & 6.33 & 2.1 & 8 \\
N-PACE (intent comm.) & 4.00 & 1.5 & 3 \\
\bottomrule
\end{tabular*}
\caption{3-vehicle platooning results. Reported error-convergence metrics are for the worst-performing vehicle among the three.}
\label{tab:platooning3}
\end{table}





\end{document}